\documentclass{llncs}
\usepackage{amsmath}
\usepackage{amssymb}
\usepackage[ruled,noend,linesnumbered]{algorithm2e}
\DontPrintSemicolon

\newcommand{\coeff}{\mathrm{coeff}}
\newcommand{\DFT}{\mathrm{DFT}}
\newcommand{\poly}{\mathrm{poly}}
\newcommand{\SP}{\mathrm{span}}
\newcommand{\FF}{\mathbb{F}}
\newcommand{\ZZ}{\mathbb{Z}}
\newcommand{\CC}{\mathbb{C}}
\newcommand{\FG}{{\FF[G]}}
\newcommand{\GF}{\mathrm{GF}}
\newcommand{\abs}[1]{\lvert #1 \rvert}

\newcommand{\weff}{w_{\mathrm{eff}}}

\begin{document}
\title{Finding the Minimum-Weight $k$-Path}



\author{Avinatan Hassidim\thanks{Research supported by ISF grant 1241/12 and by GIF Young.} \and Orgad Keller \and Moshe Lewenstein\thanks{Research supported by BSF grant 2010437, a Google Research Award and GIF grant 1147/2011.} \and Liam Roditty}

\institute{Department of Computer Science, Bar-Ilan
University, Ramat-Gan, Israel \email{\{avinatan,kellero,moshe,liamr\}@cs.biu.ac.il}}

\maketitle

\begin{abstract}
Given a weighted $n$-vertex graph $G$ with integer edge-weights taken from a range $[-M,M]$, we show that the minimum-weight simple path visiting $k$ vertices can be found in time $\tilde{O}(2^k \poly(k) M n^\omega) = O^*(2^k M)$. If the weights are reals in $[1,M]$, we provide a $(1+\varepsilon)$-approximation which has a running time of $\tilde{O}(2^k \poly(k) n^\omega(\log\log M + 1/\varepsilon))$. For the more general problem of $k$-tree, in which we wish to find a minimum-weight copy of a $k$-node tree $T$ in a given weighted graph $G$, under the same restrictions on edge weights respectively, we give an exact solution of running time $\tilde{O}(2^k \poly(k) M n^3) $ and a $(1+\varepsilon)$-approximate solution of running time $\tilde{O}(2^k \poly(k) n^3(\log\log M + 1/\varepsilon))$. All of the above algorithms are randomized with a polynomially-small error probability.
\end{abstract}

\section{Introduction}
Given an $n$-vertex graph $G=(V,E)$ and a parameter $k$, in the \emph{$k$-path} problem we wish to find a path in $G$ consisting of $k$ vertices, if such exists. 
The $k$-path problem can be easily shown to be NP-complete: when $k=n$, it is exactly the Hamiltonian path problem. While a trivial $O^*(n^k)$ solution\footnote{Here and throughout, the $O^*$ notation discards all factors that are polynomial in $n$, $k$, and $\log M$ from the running time. Similarly, the $\tilde{O}$ expressions discard poly-logarithmic factors. 
} is to try all $\binom{n}{k}$ combinations of $k$ vertices, better can be obtained; Monien was the first to show an improvement~\cite{Monien85}, with an $O^*(k!)$ algorithm. In their seminal result, Alon, Yuster, and Zwick~\cite{AYZ95}  introduced the \emph{color-coding} technique. They used it to present a randomized $O^*((2e)^k)$ algorithm for this problem, which can be derandomized, replacing the $2e$ term with a large constant. Their result thus shows that the \textsc{logpath} problem of determining whether a graph contains a path of length $\log n$ can be solved in polynomial time.
Later, two independent results~\cite{KMRR06,CLSZ07} presented randomized $O^*(4^k)$ algorithms, again with larger constants when derandomized, having running times of $O^*(16^k)$~\cite{KMRR06} and $O^*(12.5^k)$~\cite{CLSZ07}.

While these results were combinatorial in nature, the next improvements used algebraic techniques: Koutis~\cite{Koutis08} presented an algorithm solving the problem in $O^*(2.83^k)$ time. His method was perfected by Williams~\cite{Williams09}, reducing the running time to $O^*(2^k)$. This had somewhat closed the gap between the $k$-path problem and the best method known for the specific case of finding a Hamiltonian path in a directed graph, which is $O^*(2^n)$ (though the latter is combinatorial in nature). For undirected graphs, recent results presented $O^*(1.657^n)$~\cite{Bjorklund10} and later $O^*(1.657^k)$~\cite{BHKK10,AB13}
running times for Hamiltonian path and $k$-path, respectively.

It is worthwhile to focus on Koutis' and Williams' techniques, as they are the basis to this paper. They reduce $k$-path and other problems to the problem of determining whether a given $n$-variable polynomial contains a $k$-multilinear-monomial (that is, a term which is the multiplication of $k$ distinct variables) in its sum-product expansion. The problem is then solved by (roughly) evaluating this polynomial over random values taken from an adequate choice of an algebraic structure. 
In a later result~\cite{KW09} they both show that, in the evaluation framework they use, their technique for finding a $k$-multilinear-monomial is essentially optimal, as any choice of an algebraic structure for the polynomial evaluation would require that the elements in this structure have an $\Omega(2^k/k)$-sized representation. 

One of the most natural generalizations coming to mind, is the \emph{minimum-weight $k$-path} problem: in this scenario, the graph edges are weighted and we wish to find a $k$-path having minimum weight in the graph. In~\cite{Williams09} this was referred to as the \emph{short cheap tour} problem and mentioned that while the $O^*(4^k)$ methods can be easily extended to accommodate for this version, the algebraic methods do not seem to support such extension, and left this as an open problem. We solve this problem for the specific case in which the edge weights are integers in the range $[-M,M]$, incurring a running time which also has a superlinear dependency on $M$. If the weights are reals in $[1,M]$ (or can be normalized to this range, as is the case if they are in the range $[\ell,h]$ for $0 < \ell < h$), we provide a $(1+\varepsilon)$-approximation which reduces this dependency to $\log \log M$. Notice that by this we conform to the important line of research in recent years, of discussing variants of distance problems in which edge-weights are integers taken from a bounded range, see e.g.,~\cite{Zwick02,CGS12}.

Another problem that generalizes $k$-path is presented in~\cite{KW09}: in the $k$-tree problem, given an $n$-vertex graph $G$ and a $k$-node tree $T$, find a copy of $T$ in $G$. For a similar generalization of this problem to \emph{minimum-weight $k$-tree}, and under similar restrictions on the edge weights, we show similar exact and approximate results.
\paragraph{Paper Organization.}
In Section~\ref{sec:method}, we first present an $\tilde{O}(2^k \poly(k) M n^\omega )$ algorithm for computing the weight of the minimum-weight $k$-path when edge weights are integers in $[-M,M]$, where $\omega < 2.3727$ stands for the matrix multiplication exponent~\cite{Williams12}. In Section~\ref{sec:actual}, we show how to find the path itself, incurring an $O(k \cdot \poly\log  n)$ multiplicative overhead for the above algorithm. Finally, in Section~\ref{sec:approx}, for the case of real edge-weights in $[1,M]$, we provide a $(1+\varepsilon)$-approximation algorithm that reduces the dependency on $M$ to $\log\log M$, by using a technique of careful adaptive scaling of the edge weights. The overall running time of this algorithm is $\tilde{O}(2^k \poly(k) n^\omega(\log\log M + 1/\varepsilon))$. 

In Section~\ref{sec:tree} we turn to the $k$-tree problem, and show similar results: we present an $\tilde{O}(2^k \poly(k) M n^3)$ algorithm for finding the minimum-weight $k$-tree when edge weights are integers in $[-M,M]$, and for the case the edge-weights are reals in $[1,M]$, provide a $(1+\varepsilon)$-approximation algorithm having running time $\tilde{O}(2^k \poly(k) n^3(\log\log M + 1/\varepsilon))$.

\section{Preliminaries}\label{sec:prelim}
We follow Williams' notation~\cite{Williams09}. Let $\FF$ be a field and $G$ be a multiplicative group. The group algebra $\FF[G]$ is defined over the set of elements of the form
\begin{equation}
\sum_{g \in G}a_g g
\end{equation}
where $a_g \in \FF$ for all $g \in G$, i.e., on the set of sums of elements from $G$ with coefficients from $\FF$. 
Addition is computed component-wise as
\begin{equation}
\sum_{g \in G}a_g g + \sum_{g \in G}b_g g = \sum_{g \in G}(a_g + b_g) g \enspace,
\end{equation}
multiplication is defined in the form of a convolution:
\begin{equation}
\left(\sum_{g \in G}a_g g\right)  \left(\sum_{g \in G}b_g g\right) = \sum_{g,h \in G}a_g b_h gh = \sum_{g \in G} \left(\sum_{h \in G}  a_h b_{h^{-1} g} \right) g \enspace,
\end{equation}
(since $G$ is a multiplicative group, the expression $h^{-1} g$ here replaces the expression of the type $g-h$ which is usually found in a convolution definition) and multiplication by a scalar $c \in \FF$ as
\begin{equation}
c\left(\sum_{g \in G}a_g g\right) = \sum_{g \in G} c a_g g \enspace.
\end{equation}

Let $0_\FF, 1_\FF$ be the addition and multiplication identities of $\FF$, respectively. Let $1_G$ be the identity of $G$. It is easy to verify that $\FF[G]$ is a ring where the addition identity element $0_{\FF[G]}=\sum_{g \in G}0_\FF \cdot g$ is the element having all coefficients taken as $0_\FF$, and the multiplication identity element $1_{\FF[G]} = 1_\FF \cdot 1_G = 1_G$. For ease of notation, hereafter $0$ and $1$ will denote $0_{\FF[G]}$ and $1_{\FF[G]}$, respectively.

Let $z$ be a symbolic variable. Our computations are done on the set $(\FG)[z]$ of univariate polynomials on $z$ with coefficients in $\FF[G]$. Notice that the set of polynomials with coefficients in a ring is a ring by itself.

For our algorithm, we follow Williams and choose $G$ to be $\ZZ_2^k$ (i.e., the set of binary vectors of dimension $k$) with multiplication between elements of $\ZZ_2^k$ defined as entry-wise addition modulo $2$. It follows that $1_G$ is the $k$-dimensional all-zeros vector. Notice that for all $u,v \in \ZZ_2^k$, $u \cdot v = 1_G$ iff $u = v$. We also choose $\FF=\GF(2^\ell)$ for $\ell = \log k + 3$. Notice that since $\FF = \GF(2^\ell)$ has characteristic $2$, it holds that for all $c \in \FF$, $c + c = 0_\FF$, and therefore that for all $v \in \FG$, $v + v = 0$.

\section{Method}\label{sec:method}
Given a weighted, directed or undirected graph $H=(V,E,w)$ on the vertex-set $V = \{1,\ldots,n\}$, with integer edge-weights in $[-M,M]$, we first show how to compute the weight of the minimum-weight $k$-path with high probability. We can assume the edge weights are actually in $[0,M]$, otherwise we re-define $w(i,j) \gets w(i,j) + M$ for each $(i,j) \in E$ and then $M \gets 2M$: as this process incurs a penalty of $(k-1)M$ for each $k$-path, it maintains the order relation on $k$-path weights.
Define a \emph{$k$-walk} to be a walk in the graph comprised of $k$ (not necessarily distinct) vertices, and let $I = \langle i_1,\ldots,i_k \rangle$ be some arbitrary $k$-walk in $H$. With a slight abuse of notation, we will also use $I$ to denote the set of edges participating in the walk.

We define a collection $\{B_c\}_{c=1}^{k-1}$ of polynomial matrices $B_c$ as follows: 
\begin{equation} B_c[i,j]=\begin{cases}
y_{i,j,c} \cdot x_i \cdot z^{w(i,j)} & \text{if $(i,j) \in E$,}\\
0 & \text{if $(i,j) \notin E$;}\\
\end{cases} \end{equation} 
where each variable $y_{i,j,c}$ shall be assigned with a randomly selected value from $\FF $ and each $x_i$ will be assigned with a value chosen from $\FG$ by a method to be described shortly. Notice that each $x_i$ corresponds to vertex $i$. Assume the values $\{y_{i,j,c}\}_{{i,j,c}}$ have already been chosen. Recall that $z$ is a symbolic variable. We define the polynomial $P$ as follows: $P(x_1, \ldots, x_n, z) = \vec{1} \cdot  B_1  \cdots B_{k-1} \cdot\vec{x}$, where $\vec{1}$ is the $n$-dimensional all-ones vector and $\vec{x}$ is the vector $(x_1, \ldots, x_n)$. 
Re-writing $P$ as its sum-product expansion we get:
\begin{equation}
P(x_1, \ldots, x_n, z) = \sum_{\substack{I\\I=\langle i_1,\ldots,i_k \rangle \text{ is a walk in }H}}\left(\prod_{c=1}^{k-1} B_c[i_c,i_{c+1}]\right)x_{i_k}\enspace,
\end{equation}
that is, $P$ is an aggregate sum over all $k$-walks in $H$, where each walk $I=\langle i_1,\ldots,i_k \rangle $ is represented by the product of its corresponding components in $B_1,\ldots, B_{k-1}$, finally multiplied by $x_{i_k}$ which corresponds to the final vertex of the walk. By substituting the $B_c[i_c,i_{c+1}]$'s for their values, and re-arranging the walk's product such that the $y_{i,j,c}$ terms appear first, then the $x_i$ terms, and finally the $z$ term, it follows that
\begin{equation}
P(x_1, \ldots, x_n, z) = \sum_{\substack{I\\I=\langle i_1,\ldots,i_k \rangle \text{ is a walk in }H}} y^I \cdot x^I \cdot z^{w(I)}\enspace,
\end{equation}
where $y^I = \prod_{c=1}^{k-1} y_{i_c, i_{c+1}, c}$, $x^I = x_{i_1} \cdots x_{i_k}$, and $w(I) = \sum_{e \in I}w(e)$ is the weight of walk $I$.
\subsection{Algorithm}
Given $H$, randomly choose all values $y_{i,j,c} \in \FF$, and randomly pick $n$ vectors $v_1,\ldots,v_n$ from $G=\ZZ_2^k$. Now compute the polynomial $P'(z)=P(1_G+v_1,\ldots,1_G+v_n,z)$. Let $\coeff_z^d P'(z)$ be the $d$-th degree term coefficient of $P'(z)$, and let $d'=\min\{d \mid \coeff_z^d P'(z) \text{ is not $0$} \}$ (if such exists). If $d'$ exists, return it. Otherwise output ``no $k$-path exists in $H$''.

\subsection{Proof of Correctness}
Let $I$ be the minimum-weight $k$-simple-path in $H$, and notice that $w(I)$ is represented in $P$ by the term $z^{w(I)}$ in the product corresponding to $I$. Notice that while no degrees $d < w(I)$ occur in $P$, it might be that the $w(I)$-th degree term of $P$ was eliminated when (partially) evaluating $P$. Our goal is to show that this happens with low probability. For a walk $I$, notice that if $I$ is simple, i.e., it visits every node at most once, then $x^I$ is multilinear, or equivalently, square-free, since each variable $x_i$ appears in it at most once. On the other hand, if $I$ is non-simple, then $x^I$ must contain some square $x_j^2$. Therefore, in order to prove the algorithm correct, we need to show that w.h.p., (a) products corresponding to non-simple paths vanish, (b) products corresponding to simple-$k$-paths do not vanish by their evaluation, and that (c) products corresponding to simple-$k$-paths are not eliminated when they are summed with other (same-degree) products. 

These notions are captured by the following propositions, which are similar to the ones in~\cite{Williams09}. 
Due to lack of space and for completeness, proofs are detailed in the appendix.

\begin{proposition}\label{pro:vanish}
If $x^I$ is non-multilinear, it vanishes.
\end{proposition}

Let $J=\sum_{v \in G}v$ be the sum of all vectors from $G=\ZZ_2^k$ (addition here is the addition of $\FG$).

\begin{proposition}\label{thr:independent}
Let $I = \langle i_1,\ldots,i_k \rangle$ be a $k$-walk. If $x^I$ is multilinear (i.e., $I$ is a $k$-path), then if the vectors $v_{i_1}, \ldots , v_{i_k} \in \ZZ_2^k$ are linearly independent w.r.t.\ entry-wise addition modulo $2$, then $x^I = J$.
\end{proposition}

\begin{corollary}\label{cor:2}
Let $I = \langle i_1,\ldots,i_k \rangle$ be a $k$-walk. If $x^I$ is multilinear (i.e., $I$ is a $k$-path), then with probability at least $0.28$ it does not vanish.
\end{corollary}

We have shown that with at least constant probability, multilinear terms do not vanish when they are assigned values as described. However, it still might happen that such multilinear terms will get eliminated when they are summed up with other multilinear terms. The next two propositions show that this can happen with at most constant probability.

\begin{proposition}\label{thr:dependent}
Let $I = \langle i_1,\ldots,i_k \rangle$ be a $k$-walk. If the variables $v_{i_1}, \ldots , v_{i_k} \in \ZZ_2^k$ are linearly dependent w.r.t.\ entry-wise addition modulo $2$, then $x^I$ vanishes.
\end{proposition}

Recall that $P(x_1, \ldots, x_n, z)$ is a polynomial in $z$ and therefore can be viewed as
\begin{equation}
P(x_1, \ldots, x_n, z) = \sum_{d=0}^{kM}\sum_{\substack{I\\I=\langle i_1,\ldots,i_k \rangle \text{ is a walk in }H\\w(I)=d}} y^I \cdot x^I \cdot z^{d}\enspace.
\end{equation}
It is therefore easy to see that the minimum-degree term in $P$ corresponds to minimum-weight $k$-paths in $H$. Let $d'$ be the minimum degree of $P$ and let 
\begin{equation}
\coeff_z^{d'} P(x_1, \ldots, x_n, z) = \sum_{\substack{I\\I \text{ is a walk in }H\\w(I)=d'}} y^I \cdot x^I
\end{equation}
 be its corresponding coefficient.
Our goal is to show that with at least constant probability, $\coeff_z^{d'} P$ does not vanish when it is evaluated.
\begin{proposition}\label{pro:high-pro}
$\coeff_z^{d'} P'(z)$ does not vanish with probability at least $1/5$.
\end{proposition}
\subsection{Running Time Analysis}
The running time of the algorithm is dominated by $k$ matrix multiplications, where the basic arithmetic operations are done over the polynomial ring $(\FG)[z]$.  Therefore, we need to account for the the cost of each such operation. Notice that for any arithmetic operation in $(\FG)[z]$ performed by our algorithm, the maximum degree of the operand polynomials and resulting polynomial, is at most $kM$. 
We can therefore focus on the set $R$ of polynomials in $(\FG)[z]$ with degree at most $kM$. By treating the polynomials in $R$ as periodic with period $kM$ (since there will be no carry or overflow to greater degrees), $R$ continues to be a ring. 
Let $T$ be the upper-bound on the time required for an arithmetic operation in $R$; trivially, $T = \Omega(2^k\cdot kM \log\abs{\FF})$. It follows that the algorithm requires $O(k n^\omega T)$ time, and it remains to compute $T$.

\paragraph{Addition.} Addition of two polynomials can be easily done component-wise in time $O(kM \cdot 2^k \cdot \log \abs{\FF}) = O(2^k \poly(k) M)$. 
\paragraph{Multiplication.} Multiplication is trickier and is done by employing a multidimensional fast Fourier transform-type approach.\footnote{Here, as opposed to Williams~\cite{Williams09}, the Walsh-Hadamard transform is not an adequate choice anymore due to the existence of the variable $z$ which can have a degree up to $kM$. } We now describe the multiplication process in more detail.

The multiplication process will be easier to describe on the ring $\FF[\ZZ_2^k\times [kM]]$ which is isomorphic to $R$, as will be shown immediately. Given a vector $v =(v_1,\ldots,v_k) \in \ZZ_2^k$ and an integer $d \in [kM]$, let $(v;d)$ denote the vector $(v_1,\ldots,v_k,d)\in \ZZ_2^k\times [kM]$. A polynomial $p \in R$ can be uniquely described as a sum $\sum_{v,d}a_{(v;d)}\cdot(v;d)$ of at most $N = 2^k kM$ summands, where each $a_{(v;d)} \in \FF$ is the coefficient of $v$ appearing in $\coeff_z^d p$ (i.e.,  if $\coeff_z^d p = \sum_{v \in G} b_v v$, then $a_{(v;d)} = b_v$). Our definition of multiplication over $G=\ZZ_2^k$ can be naturally extended to $\ZZ_2^k\times [kM]$: multiplication still corresponds to entry-wise addition, only that now addition is done modulo $2$ for dimensions $1,\ldots,k$ and modulo $kM$ for dimension $k+1$. With that in mind, our definitions of addition, multiplication, and identity elements for $R$ are extended appropriately, thus forming the ring $\FF[\ZZ_2^k\times [kM]]$. The bottom line is that now any $p \in R$ can be viewed as a sum of elements with coefficients taken from a multidimensional array indexed by values from $\ZZ_2^k\times [kM]$ and that multiplication is still a convolution, an important fact to be used later.

Moving to $\FF=\GF(2^\ell)$, being a finite field, all elements in $\FF$ can be represented in the usual manner as a degree-$\ell$ polynomials with coefficients in $\ZZ_2=\GF(2)$ and operations that are done modulo some predefined irreducible polynomial of degree $\ell$ (this irreducible polynomial can even be found na\"{\i}vely as $\ell=\log k + 3$). For the purpose of using FFT,  we treat polynomials in $\ZZ_2[x]$ as if they were actually in $\CC[x]$, i.e., the set of univariate polynomials over the complex numbers. At the end of the multiplication process, we will appropriately convert polynomials in $\CC[x]$ back to $\GF(2^\ell)$ as will be described shortly.

By the above arguments, given two polynomials $p,q \in R$ to be multiplied, they can be taken as the sums $\sum_{v,d}p_{(v;d)}\cdot(v;d)$ and $\sum_{v,d}q_{(v;d)}\cdot(v;d)$, respectively, where $p_{(v;d)}, q_{(v;d)} \in \CC[x]$ for each $v\in \ZZ_2^k$ and $d\in [kM]$. As the multiplication corresponds to a convolution, by the convolution theorem, it holds that $p * q = \DFT^{-1}(\DFT(p) \cdot \DFT(q))$, where $*$ denotes a convolution, $\cdot$ denotes point-wise multiplication, and $\DFT$ denotes the $(k+1)$-dimensional discrete Fourier transform for values indexed by vectors of type $(v_1,\ldots,v_k,d)\in \ZZ_2^k\times [kM]$. 
Let $D(\ell)$ denote the time required for an arithmetic operation on degree-$\ell$ polynomials in $\CC[x]$---including converting them back to $\GF(2^\ell)$ by division by an irreducible polynomial---and notice that $D(\ell) = O(\ell^2) = O(\poly\log k)$ as multiplication and division here are quadratic by nature.
Then the above DFT operations can be computed efficiently in time $O(N \log N \cdot D(\ell)) = \tilde{O}(2^k k^2 M)$ by using the multidimensional FFT algorithm.
Once we have computed $\DFT(p)$ and $\DFT(q)$, thus obtaining for each of them $N$ values in $\CC[x]$ (indexed as well by vectors in $\ZZ_2^k\times [kM]$), we point-wise multiply them, obtaining a sum $w = \DFT(p) \cdot \DFT(q)$, and compute $\DFT^{-1}(w)$, again by using FFT on multidimensional coefficients in $\CC[x]$. Finally, we reduce $\CC[x]$ terms (which are actually in $\ZZ[x]$, as convolution over integer values returns integer values) by dividing them by the irreducible polynomial used before and the appropriate modulo operations. 
 
We conclude that multiplication of polynomials in $R$ can be performed in time $\tilde{O}(2^k \poly(k) M )$, and therefore $T = \tilde{O}(2^k \poly(k) M)$.

\section{Finding the Actual Path}\label{sec:actual}
Let $G=(V,E,w)$ be a weighted graph with integer edge-weights in $[-M,M]$. Given the algorithm from the previous section, we show that it is possible to find the minimum-weight $k$-path itself with only $O(k \poly\log n)$ multiplicative overhead w.r.t.\ the previous algorithm and with a polynomially small error probability. 
We denote by $\mathcal{A}$ the algorithm from the previous section, amplified by running $O(\log n)$ iterations of it and choosing the minimal result, such that its error probability is bounded by $1/n^{c'}$ for some constant $c'$. 
The algorithm for finding the actual path uses $\mathcal{A}$ as a sub-routine. Its pseudo-code is provided as Algorithm~\ref{alg:finding}. 
The rest of this section is deferred to Section~\ref{sec:actual:app} of the appendix due to lack of space.

\begin{algorithm}[t]\label{alg:finding}
\caption{Finding the minimum-weight $k$-path.}
$d \gets \mathcal{A}(G,k)$\;
\While{$\abs{V(G)} > 10k$}{
\For{$\Theta(\log n)$ times} 
	{
	$G' \gets $ a copy of $G$ in which each vertex is removed with probability $1/k$\;
	\If{at least $\Omega(\abs{V(G)}/k)$ were removed and $\mathcal{A}(G',k)=d$}{
	$G \gets G'$\;
	Go to the while loop\;
	}
	
	}
	\Return{``Fail''}
}
\ForEach{remaining vertex $v \in V(G)$ and until $\abs{V(G)}=k$}{
	$G' \gets G \setminus v$\tcc*{$G \setminus v$ is $G$ with $v$ and its incident edges removed}
	\lIf{$\mathcal{A}(G',k)=d$}{$G \gets G'$}
}
\Return{$E(G)$}

\end{algorithm}

\section{Approximation}\label{sec:approx}
The main drawback of the previous algorithm is that its running time has a superlinear dependency in $M$, the bound on an edge weight. If the weights are in $[1,M]$ (or can be normalized to this range), we show that if we settle for a $(1+\varepsilon)$-approximation algorithm to the problem, this dependency can be brought down to $\log\log M$, by using a technique of careful adaptive scaling of the edge weights, thus bringing the overall running time to $\tilde{O}(2^k \poly(k) n^\omega(\log\log M + 1/\varepsilon))$. Our techniques are in the spirit of the FPTAS of Erg{\"u}n et al.~\cite{ESZ02} for the restricted shortest path problem.  
We start with the following proposition:

\begin{proposition}\label{lemma:bound}
Given a graph $G$ with integer edge-weights in $[0,M]$, a parameter $k$, and a value $B$, it is possible to find an exact solution to the minimum-weight $k$-path problem of weight at most $B$, if such exists, or to return that no such solution exists, in time $\tilde{O}(2^k \poly(k) B n^\omega) = O^*(2^k B)$ and polynomially-small error probability.\footnote{$B$ does not have to be an integer, but the effect in this case is as if $\lfloor B \rfloor$ is used.}
\end{proposition}
\begin{proof}
The algorithm is identical to the previous one, except that as a first step, edges of weight greater than $B$ are deleted from the graph, and that when multiplying two polynomials in $(\FG)[z]$ of degree at most $B$, we truncate from the resulting polynomial any term of degree greater than $B$, thus keeping all polynomials throughout the algorithm at degree of at most $B$. As every polynomial multiplication now takes $\tilde{O}(2^k \poly(k) B)$ time, the running time analysis follows.
\qed\end{proof}
We denote with $\mathcal{B}$ the algorithm that finds an exact solution to the $k$-path problem of weight at most $B$, if such exists, or to returns that no such solution exists. We will use it as a sub-routine in our approximation algorithm.

Define $k'=k-1$ (the number of edges in a $k$-path), and let $OPT$ be the minimum-weight $k$-path. Our approximation algorithm starts by defining an upper and a lower bound, $U$ and $L$, respectively, to the weight of $OPT$. At first, $U=k'M$ and $L=k'$. It then iteratively fine-tunes $U$ and $L$ to the point where the ratio $U/L$ is less than or equal to $2$, while maintaining the invariant that $L \leq w(OPT) \leq U$. This fine tuning is done as follows. 

At each iteration we let the value $X=\sqrt{LU}$ be the geometric mean of $L$ and $U$, and define the value $\delta = (L/U)^{1/3}-\sqrt{L/U}$ which will serve as a scaling coefficient. Notice that $\delta > 0$ as $U>L$. We then scale-down the edge weights by a factor of $\delta U / k'$, thus defining a new weight $w'(i,j) = \left\lfloor{\frac{w(i,j)}{\delta U /k'}}\right\rfloor$ for each edge $(i,j)$, and let $G'=(V,E,w')$ be the graph with the new weights. 
Ideally, we would like to test whether the weight of the optimal solution is less than or greater than $X$ by calling $\mathcal{B}(G', k,   \frac{X}{\delta U/k'})$; here notice that the value $\frac{X}{\delta U/k'}$ is the scaled-down equivalent of $X$ in $G'$.
However, while the scaling guarantees that this test can be done without incurring a high running time cost, it also introduces a loss of precision due to the floor function in the scaling: define $\weff(i,j)=(\delta U /k')w'(i,j)$ as the effective weight $w'(i,j)$ simulates, then we have that $\weff(i,j) \leq w(i,j) \leq \weff(i,j) + \delta U /k'$, and therefore for a $k$-path $P$, we have that $\weff(P) \leq w(P) \leq \weff(P) + \delta U$. Therefore, in the case $w'(OPT) >   \frac{X}{\delta U/k'}$ we have that $w(OPT) \geq \weff(OPT) > X$, but if $w'(OPT) \leq   \frac{X}{\delta U/k'}$ (and therefore $\weff(OPT) \leq X$) then all we can assert is that $w(OPT) \leq X + \delta U$.
Therefore, a $k$-path returned by a call to $\mathcal{B}(G', k,   \frac{X}{\delta U/k'})$ has weight at most $X + \delta U$ (and not $X$) w.r.t.\ the original graph. According to the outcome of the call to  $\mathcal{B}(G', k,   \frac{X}{\delta U/k'})$, we redefine $U$ and $L$: if $\mathcal{B}(G', k,   \frac{X}{\delta U/k'})$ returned a result, we set $U \gets X+\delta U$; otherwise we set $L \gets X$. 

When the main loop is done (convergence is shown to exist below), we again redefine a new weight function: $w'(i,j) = \left\lfloor{\frac{w(i,j)}{\varepsilon L /k'}}\right\rfloor$ for each edge $(i,j)$, the graph $G'=(V,E,w')$, and return the result of a call to $\mathcal{B}(G',k,  \frac{U}{\varepsilon L/k'})$. The full algorithm pseudo-code is given as Algorithm~\ref{alg:approx}.

\begin{algorithm}[t]\label{alg:approx}
\caption{Approximation algorithm.}
$k' \gets k-1$\;
$L \gets k'$\;
$U \gets k'M$  \;
\While{$U > 2L$}{
$X \gets \sqrt{LU}$\;
$\delta \gets (L/U)^{1/3}-\sqrt{L/U}$\;
Define $w' \colon E \to \mathbb{N}$ such that $w'(i,j) = \left\lfloor{\frac{w(i,j)}{\delta U /k'}}\right\rfloor$\;
$G' \gets (V,E,w')$\;
\eIf{$\mathcal{B}(G', k,   \frac{X}{\delta U/k'})$ returns a result}{$U \gets X+\delta U$\;}{$L \gets X$\;}
}
Define $w' \colon E \to \mathbb{N}$ such that $w'(i,j) = \left\lfloor{\frac{w(i,j)}{\varepsilon L /k'}}\right\rfloor$\;\label{line:weights-final}
$G' \gets (V,E,w')$\;
\Return{$\mathcal{B}(G',k,  \frac{U}{\varepsilon L/k'})$}\label{line:return}
\end{algorithm}

\paragraph{Running-Time.} 
We first show that the main loop performs $O(\log\log M)$ iterations. Let $L_i,U_i$ be the respective values of $L,U$ at the start of iteration $i$; we will show that $U_{i+1}/L_{i+1} \leq (U_{i}/L_{i})^{2/3}$. At the end of each iteration $i$, we have that either $L_{i+1}\gets L_i$ and $U_{i+1}\gets X+\delta U_i$, or that $L_{i+1}\gets X$ and $U_{i+1}\gets U_i$, where $X=\sqrt{L_iU_i}$ and $\delta = (L_i/U_i)^{1/3}-\sqrt{L_i/U_i}$. In the former case we have that 
\begin{equation}
\frac{U_{i+1}}{L_{i+1}}=\frac{X+\delta U_i}{L_i}=\frac{\sqrt{L_iU_i} + \left(\left(\frac{L_i}{U_i}\right)^{1/3}-\sqrt{\frac{L_i}{U_i}}\right)U_i}{L_i} \\ =\frac{\left(\frac{L_i}{U_i}\right)^{1/3}U_i}{L_i}=\left(\frac{U_i}{L_i}\right)^{2/3}\enspace,
\end{equation}
and in the latter
\begin{equation}
\frac{U_{i+1}}{L_{i+1}}=\frac{U_i}{X}=\frac{U_i}{\sqrt{L_iU_i}}=\sqrt{\frac{U_i}{L_i}} \leq \left(\frac{U_i}{L_i}\right)^{2/3}\enspace.
\end{equation}
In both cases we have that $U_{i+1}/L_{i+1} \leq (U_{i}/L_{i})^{2/3}$. Therefore it converges to a constant after $O(\log\log M)$ iterations. Notice that an invocation of $\mathcal{B}(G', k,  \frac{X}{\delta U/k'})$ costs $\tilde{O}(2^k \poly(k) n^\omega)$ by Proposition~\ref{lemma:bound}, with the bound $B=  \frac{X}{\delta U/k'}$ which is $O(k)$, as $\delta U = \Omega(X)$. We conclude that the overall cost of the main loop is $\tilde{O}(2^k \poly(k) n^\omega \log\log M)$.

As for the final call to $\mathcal{B}(G',k,  \frac{U}{\varepsilon L/k'})$, we have that its running time is $\tilde{O}(2^k \poly(k) n^\omega /\varepsilon)$ by Proposition~\ref{lemma:bound}, with the bound $B=  \frac{U}{\varepsilon L/k'}$ which is $O(k/\varepsilon)$ since at this stage $U \leq 2L$. We conclude that the overall running time of the approximation algorithm is $\tilde{O}(2^k \poly(k) n^\omega(\log\log M + 1/\varepsilon))$.

\paragraph{Correctness.} 
Throughout the execution, the algorithm maintains the invariant that $L < X < X+\delta U< U$. That can be easily seen by substituting $X$ and $\delta$ for their values and observing that $L < \sqrt{LU} < L^{1/3}U^{2/3} < U$. 
Assume there exist a $k$-path in $G$, and let $OPT$ be the minimum-weight $k$-path. By the scaling arguments, and the fact that we have brought the loss of precision due to scaling into consideration when redefining $U$ and $L$, we have that the invariant $L \leq w(OPT) \leq U$ always holds. Due to the running-time argument, when the main loop is done we have $U/L \leq 2$. Let $P^*$ be the result of the call to  
$\mathcal{B}(G',k,  \frac{U}{\varepsilon L/k'})$ 
at line~\ref{line:return} of the pseudo-code, and notice the the weights defined at line~\ref{line:weights-final} incur an $\varepsilon L /k'$ loss of precision per edge, or equivalently $\varepsilon L $ per $k$-path. By the call to the exact algorithm, we have that $w'(P^*) \leq w'(OPT)$ 
and therefore also $\weff(P^*) \leq \weff(OPT)$. Accounting for the loss of precision, we have that $w(P^*) \leq \weff(P^*) + \varepsilon L \leq \weff(OPT) + \varepsilon L \leq (1+\varepsilon)w(OPT)$.

\section{$k$-tree}\label{sec:tree}
In~\cite{KW09}, they provide a solution to the $k$-tree problem: given an $n$-vertex graph $G$ and a $k$-node tree $T$, is there a (not necessarily induced) copy of $T$ in $G$. Again their solution is based on a reduction to the question of is there a $k$-multilinear-monomial in the sum-product expansion of a given polynomial. We show how to handle the \emph{minimum-weight $k$-tree} problem---in which we are given a weighted graph $G$, and wish to find a minimum-weight copy of $T$ in it, across all copies of $T$ in it---again, when the weights are integers in a given range $[-M,M]$. 

\begin{theorem}\label{thr:tree}
Given a graph $G$, if the edge-weights are integers in $[-M,M]$, the minimum-weight $k$-tree can be found in $\tilde{O}(2^k \poly(k) M n^3)$ time. If the edge-weights are reals in $[1,M]$, the problem can be approximated within $(1+\varepsilon)$ in $\tilde{O}(2^k \poly(k) n^3(\log\log M + 1/\varepsilon))$ time.
\end{theorem}
Let $N_G(i)$ be the neighbor-set of vertex $i$ in $G$, and let $X=\{x_1,\ldots,x_n\}$ be a variable-set corresponding to $V(G)$. We use the following polynomial on $X$, implemented as an arithmetic circuit:

Let $V(G)=[n]$ and $V(T)=[k]$. The polynomial $C_{T,i,j}(x_1,\ldots,x_n)$ is defined as follows. If $\abs{V(T)}=1$, then $C_{T,i,j}=x_j$. Otherwise, $C_{T,i,j}$ is defined recursively:  let $\{T_{i,\ell} \mid \ell \in N_T(i)\}$ be the subtrees of $T$ created by removing node $i$ from $T$, where $T_{i,\ell}$ is the subtree containing $\ell$. Then 
\begin{equation}
C_{T,i,j} = \prod_{\ell \in N_T(i)}\left( \sum_{j'\in N_G(j)} y_{(i,\ell), (j,j')} \cdot z^{w(j,j')} C_{T_{i,\ell},\ell,j'} \right)\enspace,
\end{equation}
where as before, $z$ is a symbolic variable, and the values $\{y_{e, e'} \mid e \in E(T), e' \in E(G)\}$ are random values drawn from $\FF$.\footnote{In~\cite{KW09}, the $y$-values are implicit and come from the multiplication of the output of each multiplication gate with a random value taken from $\FF$.}
Finally, define the polynomial $Q = \sum_{j\in V(G)} C_{T,1,j}$. Each $C_{T,1,j}$ is a circuit containing at most $\abs{E(T)}\cdot \abs{E(G)}$ addition and multiplication gates and therefore $Q$ contains $n \cdot \abs{E(T)}\cdot \abs{E(G)} = O(n^3 k)$ such gates. 
$Q$ is a sum over all homomorphisms from $T$ to subgraphs of $G$ of size at most $k$: specifically $C_{T,i,j}$ aggregates over all homomorphisms that map $i\in V(T)$ to $j \in V(G)$ (proof can be found in~\cite{KW09}\footnote{Their arithmetic circuit is defined as $Q = \sum_{i \in V(T),j\in V(G)} C_{T,i,j}$, however, it seems to contain redundancy.}). Therefore, a monomial $x_{j_1}\cdots x_{j_k}$ appears in the sum-product expansion of $Q$ if an only if there is a homomorphism mapping $V(T)$ to $\{j_1,\ldots, j_k\}$ such that if $(i,\ell)\in E(T)$, then $(j_i, j_\ell) \in E(G)$. If such a monomial is multilinear, it corresponds to such a homomorphism in which $j_1,\ldots, j_k$ are distinct vertices, i.e., a vertex in $G$ was not used more than once for the sake of a single mapping. From this point, the same algorithms given before follow (only this time, evaluating $Q$ over $(\FG)[z]$), and propositions similar to Propositions~\ref{pro:vanish}--\ref{pro:high-pro} apply. Full proofs are deferred to the full version of the paper. We obtain that the minimum-weight $k$-tree problem with integer edge-weights in $[-M,M]$ can be solved in $\tilde{O}(2^k \poly(k) M n^3)$ time and that if the edge-weights are reals in $[1,M]$, it can be approximated within $(1+\varepsilon)$ in $\tilde{O}(2^k \poly(k) n^3(\log\log M + 1/\varepsilon))$ time.
\section{Acknowledgments}
We would like to thank Ryan Williams and Danny Raz for helpful comments.


\bibliographystyle{abbrv}
\bibliography{range}

\section*{Appendix}

\paragraph{Proof of Proposition~\ref{pro:vanish}.}
Assume $x^I$ contains some square $x_j^2$. Since $x_j$ was assigned with $1_G+v_j$, it holds that $x_j^2 = (1_G+v_j)^2 = 1_G^2 + 2 \cdot 1_G\cdot v_j + v_j^2 = 1_G + 2 \cdot  1_G \cdot v_j + 1_G = 2 \cdot  1_G+ 2 \cdot  1_G\cdot v_j = 0+0=0$ where the third equality holds since for all $v \in G$, $v\cdot v = 1_G$, and the fifth equality holds since $\FF$ has characteristic $2$ and therefore for all $c \in \FF$, $2c = 0_\FF$.\qed

\paragraph{Proof of Proposition~\ref{thr:independent}.}
If the $k$ vectors $v_{i_1}, \ldots , v_{i_k} \in \ZZ_2^k$ are linearly-independent, then they form a basis $B = \{v_{i_1}, \ldots , v_{i_k}\}$ for $\ZZ_2^k$. Notice that $x^I = \prod_{c=1}^k (1_G + v_{i_c}) = \sum_{S \subseteq B}\prod_{v \in S} v$, i.e., $x^I$ is the sum of every possible combination of vectors from $B$, multiplied together. Hence, the sum covers all vectors in the span of $B$, that is, $\sum_{S \subseteq B}\prod_{v \in S} v = \sum_{v \in \SP(B) }v = \sum_{v \in \ZZ_2^k}v = J$.
\qed

\paragraph{Proof of Corollary~\ref{cor:2}.}
The values $v_{i_1}, \ldots , v_{i_k} \in \ZZ_2^k$ were chosen randomly and independently. It is known that a random $k \times k$ matrix of values from $\ZZ_2$ has full rank with probability at least $0.28$~\cite{BK95}.
\qed

\paragraph{Proof of Proposition~\ref{thr:dependent}.}
Recall that $x^I = \sum_{S \subseteq \{v_{i_1},\ldots,v_{i_k}\}}\prod_{v \in S} v$. If the $k$ vectors $v_{i_1}, \ldots , v_{i_k} \in \ZZ_2^k$ are linearly-dependent, then there exists a set $T \subseteq \{v_{i_1}, \ldots , v_{i_k}\}$ such that $\prod_{v \in T} = 1_G$. Since, as mentioned, for $u,v \in G$ it holds that $uv = 1_G$ iff $u = v$, we get that for all $S' \subseteq T$, $\prod_{v \in S'}v = \prod_{v \in T\setminus S'}v$. It follows that every value $r = \prod_{v \in S}v$ occurs twice in the sum, one time as $r = \prod_{v \in S}v$, and one time as $r = \prod_{v \in (S \setminus T) \cup (T \setminus S)}v$. Since $2r=0_\FF \cdot r$ as $\FF$ has characteristic $2$, all terms are eliminated in the sum.
\qed

\paragraph{Proof of Proposition~\ref{pro:high-pro}.}
By Propositions~\ref{thr:independent} and~\ref{thr:dependent}, it holds that 
$$\coeff_z^{d'}P'(z) = J \cdot \sum_{\substack{I\\I \text{ is a walk in }H\\w(I)=d'\text{ and $I$ survived}}} y^I \ .$$
Let $$Q = \sum_{\substack{I\\I \text{ is a walk in }H\\w(I)=d'\text{ and $I$ survived}}} y^I \ .$$
$Q$ is a degree-$k$ polynomial in the variables $\{y_{i,j,c}\}_{i,j,c}$. With probability at least $0.28$ at least one minimum-weight $k$-path $I$ had survived and therefore $Q$ is not identically zero. In this case, by the Schwartz-Zippel lemma~\cite{Schwartz80,Zippel79,DL78}, when assigning random values from $\GF(2^\ell)$ to the variable set $\{y_{i,j,c}\}_{i,j,c}$, $Q$ evaluates to zero with probability at most $k/2^\ell=1/8$. Therefore $Q$ (and hence, $\coeff_z^{d'} P'(z)$) does not vanish with probability at least $0.28 \cdot 7/8 > 1/5$.
\qed

\subsection{Finding the Actual Path}\label{sec:actual:app}
Let $G=(V,E,w)$ be a weighted graph. We first run $\mathcal{A}(G,k)$ on the graph. Let $d$ be the value returned by it, i.e., the weight of the minimum-weight $k$-path. 

If $\abs{V} > 10k$, repeat the following procedure $\Theta(\log n)$ times:\footnote{For the sake of brevity, in this section we do not give full details of the underlying constants that are required.} remove each of the graph vertices with probability $1/k$. If $\Omega(\abs{V}/k)$ vertices were removed, run $\mathcal{A}$ on the resulting graph and $k$. If the algorithm had returned a result $d' = d$, then keep the vertices discarded indefinitely and stop, otherwise return them back to the graph. If after the $\Theta(\log n)$ iterations no vertices were discarded indefinitely, output ``Fail''. 

The above procedure is repeated as long as $\abs{V} > 10k$. Once $\abs{V} \leq 10k$, we perform an ordinary self reduction: each time we remove a different vertex and query $\mathcal{A}$ with the resulting graph and $k$; if the result stays the same, we keep this vertex discarded, otherwise, we return it to the graph. Once $\abs{V}=k$, we return the edge-set $E$ as the resulting path. This algorithm's pseudo-code is given as Algorithm~\ref{alg:finding}.

\paragraph{Error probability.} Let $P$ be the minimum-weight $k$-path in $G$, and assume $k \geq 3$, otherwise the problem is trivial.
Let $T$ be the set of vertices removed from $G$ in an iteration of the for loop. 
The probability $T$ does not include any of the vertices of $P$ is $(1-1/k)^k \geq 1/4$. Now assume it does not, in that case it holds that $E[\abs{T}] = \frac{\abs{V(G)}-k}{k} \geq \frac{9\abs{V(G)}}{10k}$, and that $Var[\abs{T}]=(\abs{V(G)}-k)(1/k)(1-1/k) < \abs{V(G)}/k$. According to Chebyshev's inequality, $\abs{T} = \Omega (\abs{V(G)}/k)$ with probability of at least a constant. It follows that the probability to pick $T$ that does not hit any of the vertices in $P$ and at the same time is $\Omega (\abs{V(G)}/k)$ is at least a constant $\alpha > 0$. 
We define this event as a ``success''. Since we perform at most $\Theta(\log n)$ trials at each iteration of the while loop, the probability of failing in all of them is $(1-\alpha)^{\Theta(\log n)}$ which can be made at most  $ 1/n^c$ for some constant $c$. By using the union-bound over the $k \ln n$ iterations of the while loop, we get a polynomially-small error probability of at most $k \ln n / n^c$. Since the probability to fail any invocation of $\mathcal{A}$ is less than $1 / n^{c'}$, by a similar union-bound argument the probability to fail in any of the calls to $\mathcal{A}$ is  $O(k \log^2 n / n^{c'})$. We obtain an overall polynomially-small error probability.

\paragraph{Running time.} Each non-failed iteration of the while loop in Algorithm~\ref{alg:finding} discards $\Omega(\abs{V(G)}/k)$ vertices and therefore reduces the number of vertices in the graph by a multiplicative factor of $(1-\Omega(1/k))$. As this happens until $\abs{V(G)} \leq 10k$, $O(k\ln n)$ iterations are enough for getting the number of vertices to $10k$.
As each iteration invoked $\mathcal{A}$ at most $O(\log n)$ times, the $O(k \poly\log n)$ multiplicative factor follows for this stage of the algorithm. As the for-each loop incurs only $O(k) < O(k \poly\log n)$ calls to $\mathcal{A}$, the running-time analysis follows.

\end{document}